\documentclass[letterpaper]{article}
\usepackage{uai2020}
\usepackage[margin=1in]{geometry}
\usepackage{times}
\usepackage{lastpage}
\usepackage{amsthm}
\usepackage{amssymb}
\usepackage{amsmath}
\usepackage{amsfonts}
\usepackage{calrsfs}
\usepackage{textcomp}
\usepackage{fancyhdr,graphicx,amsmath,amssymb}
\usepackage[ruled,vlined]{algorithm2e}
\usepackage{makecell}
\usepackage{natbib}
\usepackage{nicefrac}

\allowdisplaybreaks

\DeclareMathAlphabet{\pazocal}{OMS}{zplm}{m}{n}

\newtheorem{theorem}{Theorem}
\newtheorem{definition}{Definition}

\title{Election Control by Manipulating Issue Significance}
\author{Andrew Estornell, Sanmay Das, Edith Elkind, Yevgeniy Vorobeychik\\Computer Science \& Engineering, Washington University in St.~Louis\\\{aestornell,sanmay,yvorobeychik\}@wustl.edu\\ elkind@cs.ox.ac.uk}
\begin{document}
\maketitle
\begin{abstract}
Integrity of elections is vital to democratic systems, but it is frequently threatened by malicious actors.
The study of algorithmic complexity of the problem of manipulating election outcomes by changing its structural features is known as election control.
One means of election control that has been proposed 
is to select a subset of issues that determine voter preferences over candidates.
We study a variation of this model in which voters have judgments about relative importance of issues, and a malicious actor can manipulate these judgments.
We show that computing effective manipulations in this model is NP-hard even with two candidates or binary issues.
However, we demonstrate that the problem is tractable with a constant number of voters or issues.
Additionally, while it remains intractable when voters can vote stochastically, we exhibit an important special case in which stochastic voting enables tractable manipulation.
\end{abstract}

\section{INTRODUCTION}

Fair elections are at the core of democratic systems.
However, elections are increasingly subject to attack by malicious parties who aim to achieve personal goals at the expense of the social good~\citep{RussiaElection16}.
The problem of election vulnerability to  malicious attack has been studied in the broader literature on election control and bribery~\citep{how-hard-is-it-to-control-and-election,candidate-positioning:article,bribery:book}.
However, in much of this literature, control is exercised through a change in the election structure (e.g., adding and removing candidates), or directly the preferences of a subset of voters (bribery).
A major means of election control that has often been overlooked in the research literature is manipulation of \emph{issues} that ultimately determine voter preferences over candidates.

A recent model of election control through issue selection attempts to bridge this gap~\citep{bis:article}.
The basis of this model is the spatial theory of voting~\citep{downs:book,stv:book}, in which voters and candidates are represented as points in issue space, and a distance metric determines relative preferences, with voters preferring candidates who are similar to them on issues.
In control through issue selection, a malicious party can select a subset of issues that then determines similarity and, consequently, voter preferences.

While capturing some of the intuition about the kinds of manipulations we commonly see (through, say, the spread of misinformation and fake news), control through issue selection nevertheless misses an essential factor: what is ultimately important, and what is at the core of manipulation is the \emph{relative significance, or salience, of issues}, with issue selection being a rather extreme special case.
A recent examples that illustrates this point is
Brexit:
Until 2016, the significance of the issue of U.K. membership in the EU was comparatively negligible \citep{brexit16}.
In 2016, it became one of the central issues, with considerable evidence pointing to Russian interference as a factor~\citep{brexitRussia16,brexitRussia19}.
In general, malicious parties can impact perceptions of relative issue importance in a variety of ways.
For example, fake social media accounts can be used to coordinate widespread mentions of particular issues, increasing their salience compared to others.
Similarly, influential individuals, such as celebrities or politicians, may be willing to accept payments to be more or less vocal about particular issues. To reflect the relative difficulty or cost of these actions we limit the attacker by one of two constraints.

Our model is a significant generalization of the work of~\citet{bis:article}.
In our version, preferences of a voter over candidates are generated based on similarity in issue space, \emph{weighted by the relative importance of issues}.
We study the complexity of this problem in the context of plurality elections for two common models of voter behavior in the spatial framework: 1) deterministic voting, in which voters always vote for their most preferred candidate, and 2) stochastic voting, where the probability of a voter voting for a candidate is a monotonic function of weighted issue similarity.
We show that the control problem is in general NP-hard in either case, even with only 2 candidates.
Indeed, for the deterministic case we demonstrate hardness even with only $\Theta(\ell)$ voters, where $\ell$ is the number of issues.
Next, we exhibit several tractable special cases.
In the deterministic case, if the number of voters is $\Theta(\log(\ell))$, or the number of issues is constant, election control is in P.
In stochastic voting, in turn, control is tractable if the probability of voting for a candidate is linear in their weighted distance from the voter.

\noindent{\bf Related Work }
The complexity of controlling elections has seen extensive treatment, starting with the work of \citet{how-hard-is-it-to-control-and-election}; 
see \citet{anyone-but-him,normalized-range-control,control-bucklin-fallback,fewvoters} for further examples and the survey by \citet{bribery:book} for an overview.
Variations of this problem consider attacks that add, remove, partition or clone candidates or voters, for a variety of voting rules.
However, most of the prior election control literature considers election models in which voter preferences over candidates are given, rather than generated based on distance in issue space.
The spatial model of elections, in turn, has received considerable attention in prior literature~\citep{A-mathematical-model-of-preference-formation-in-a-democratic-society,stv:book,Advances-in-the-Spatial-Theory-of-Voting,A-decade-of-experimental-research-on-spatial-models-of-elections-and-committees,A-Unified-Theory-of-Voting-Directional-and-Proximity-Spatial-Models,Approximating-optimal-social-choice-under-metric-preferences,Randomized-social-choice-functions-under-metric-preferences}.
However, most of this research has focused on problems other than election control.
For example, extensive literature exists on game-theoretic models in which candidates opportunistically select positions in issue space~\citep{downs:book,Hotelling-Downs-Model-with-Limited-Attraction,Real-Candidacy-Games}.

A direct precursor to our model, combining election control with spatial theory of voting, is \citet{bis:article}, who study the model in which an adversary can select an arbitrary subset of issues in an election.
These issues are then used to generate voter preferences over candidates, with voters preferring candidates who are closest to them on the adversarially selected issues.
We significantly generalize this model by allowing the adversary to change relative importance of issues.


\section{MODEL}
Let $\pazocal{C} = \{\mathbf{c}_1, ..., \mathbf{c}_m \}$ and $\pazocal{V} = \{ \mathbf{v}_1, ..., \mathbf{v}_n\}$ be a set of candidates and voters, respectively. Each candidate and voter is a vector over a set of issues $[\ell]$ defined by $\mathbf{c}_i = \langle c_{i,1}, ..., c_{i, \ell}\rangle\in {\mathbb R}^\ell$ and $\mathbf{v}_j = \langle v_{j, 1}, ..., v_{j, \ell} \rangle\in {\mathbb R}^\ell$.
We consider plurality elections in which each voter is asked to report
the most preferred candidate, and the winner is the candidate who tallies the most votes.

Suppose that the relative importance of issues to voters is determined by a weight vector $\mathbf{w} = \langle w_1, ..., w_{\ell}\rangle$, where $0 \leq w_k \leq 1$ for $k\in [\ell]$, and $||\mathbf{w}||_1 = 1$.
In this model, a voter's preferences over candidates are determined by the weighted distance from each candidate in issue space.
Formally, the weighted distance between a voter $\mathbf{v}_j$ and candidate $\mathbf{c}_i$ is
\begin{align*}
d_{(i,j)} = \bigg(\sum_{k = 1}^{\ell} w_k\big|c_{i,k} - v_{j,k}\big|^p \bigg)^{1/p},
\end{align*}
$p \ge 1$, and the voter prefers a candidate who is closer according to this weighted distance measure.

Without loss of generality, suppose that $\mathbf{c}_1$ is the attacker's preferred candidate.
In our model,  an attacker aims to influence the election by modifying the relative importance of issues.
Specifically, the attacker changes $\mathbf{w}$ into a modified preference vector $\mathbf{w}' = \mathbf{w} + \mathbf{x}$. Restrictions on the attacker's strength are given in one of the following forms.
\begin{enumerate}
	\item \textbf{Normed Budget Constraint: } Given a budget $B\in \mathbb{R}$ the attacker's total perturbation must be less than the budget, i.e. $||\mathbf{x}||_p\leq B$.
	\item \textbf{Interval Constraint: } Given a set of intervals $I = I_1 \times ... \times I_{\ell} \subset [0, 1]^{\ell}$, the attacker's new weight vector must fall within $I$, i.e. $w_k \in I_k$ for all $k\in[\ell]$.
\end{enumerate} 
For both of these constraints, we consider two paradigms for voters selecting their desired candidate. 
\begin{enumerate}
    \item  \textbf{Deterministic voting:} Voter $\mathbf{v}_j$ votes for candidate $\mathbf{c}^*_i \in \text{argmin}_{\mathbf{c}_i} d_{i,j}$, breaking ties according to candidate order $\mathbf{c}_1 \prec\dots\prec\mathbf{c}_m$.
    \item \textbf{Stochastic voting:} $\mathbb{P}(\mathbf{v}_j \text{ votes for } \mathbf{c}_i) = f(\mathbf{v}_j, \mathbf{c}_i)$, where $f : \pazocal{V} \times \pazocal{C} \rightarrow [0, 1]$ is a proper probability mapping and $\sum_{i}f(\mathbf{v}_j, \mathbf{c}_i) = 1$ for all $j\in [n]$.
\end{enumerate}
Both are common models translating voter-candidate distance into voting behavior~\cite{Advances-in-the-Spatial-Theory-of-Voting}.

When voters select candidates deterministically, we consider two objectives for the attacker. The first objective (\emph{Max Support}) is to maximize the total number of votes for their preferred candidate, $\mathbf{c}_1$:
\begin{align*}
    \text{maximize}_{\mathbf{x}} \big|\{\mathbf{v}_j \in \pazocal{V}: \mathbf{c}_1 \in \text{argmin}_{\mathbf{c}_i}d_{i, j} \}\big|
\end{align*}
The second objective (\emph{Majority Vote}) is to win the plurality vote: 
\begin{align*}
    \mathbf{c}_1 \in \text{argmax}_{\mathbf{c}_i\in \pazocal{C}}|\{\mathbf{v}_j \in \pazocal{V}: \mathbf{c}_i \in \text{argmin}_{\mathbf{c}_\ell}d_{\ell, j} \}|,
\end{align*} 
For stochastic voting, we consider the objective of maximizing the expected number of votes for $\mathbf{c}_1$:
\begin{align*}
    \text{maximize}_{\mathbf{x}} \sum_{j = 1}^n f(\mathbf{v}_j, \mathbf{c}_1).
\end{align*}

\section{DETERMINISTIC VOTING}
We begin by investigating our model of election control when voters always vote for their most preferred candidate.
For compactness of notation we introduce, for every candidate and voter pair, a preference vector which gives $\mathbf{v}_j$'s unweighted preference for $\mathbf{c}_1$ over $\mathbf{c}_i$ on each issue. Let $a_{j, k}^{(i)}  = \big| c_{i,k} - v_{j,k}\big|^p- \big| c_{1,k} - v_{j,k}\big|^p$. The preference vector for $\mathbf{v}_j$ of $\mathbf{c}_1$ over $\mathbf{c}_i$ is $\mathbf{a}_j^{(i)} = \big\langle a_{j, 1}^{(i)}, ..., a_{j, \ell}^{(i)} \big\rangle$.
The condition for $\mathbf{v}_j$ voting for $\mathbf{c}_1$ is then
\begin{align*}
\sum_{k=1}^{\ell} w_k a_{j, k}^{(i)}\geq 0\quad \text{ for all } i\in [m]\setminus\{1\}.
\end{align*}
In the case of only two candidates, we omit the index $i$ and denote the preference vector for $\mathbf{c}_1$ over $\mathbf{c}_2$ by $\mathbf{a}_j$.

In this section, we will show that both \emph{Max Support} and \emph{Majority Vote} are NP-hard, even when there are only two candidates, issues are binary, and the attacker has no constraints on their strength. If these assumptions are further restricted such that there are only $\Theta(\ell)$ voters, where $\mathbf{v}_k$ agrees with $\mathbf{c}_1$ on exactly $k$ issues, then  \emph{MaxSupport} is still NP-hard. 

Although both objectives are hard, even with several strong restrictions, we present sufficient conditions for the problem to become tractable. The following positive results hold for the normed budget constraint when $p \in \{1, 2, \infty\}$, and for the interval constraint with any $p\geq 1$.  Under the normed budget, with $p\notin\{1, 2, \infty\}$, results also hold, however, the computed $\mathbf{x}$ may break the attackers budget by some small $\epsilon > 0$, i.e. for both objectives we obtain, in polynomial time, a vector $\mathbf{x}$ that is guaranteed to have $||\mathbf{x}||_p \leq B + \epsilon$. 
The tractable cases are as follows: \emph{Max Support} has a polynomial time algorithm when either there are at most $\Theta(\log(\ell))$ voters, or the number of issues and the number of values each issues can take on are both constant, and
\emph{Majority Votes} has a polynomial time algorithm when either the number of voters is constant, or the number of issues and the number of values each issues can take on are both constant.


\subsection{HARDNESS OF CONTROL IN DETERMINISTIC SETTINGS}
First we show that even for two candidates and binary issues \emph{Majority Vote} (WTCP) is NP-complete, and that \emph{Max Support} (TCWMS) is NP-hard. Both of these result are the product of hardness reductions from the problem of election control by issue selection~\cite{bis:article}. We first give a formal definition of each problem used in the reductions.
\begin{definition}
Let $\pazocal{C}'$ be a set of $m'$ candidates and $\pazocal{V}'$ be a set of $n'$ voters. Both candidates and voters are vectors in $\{0, 1\}^\ell$, indicating positions on $[\ell']$ issues. An adversary selects a subset $S \subset [\ell']$ with $S \neq \emptyset$ for the objective of either determining if  $\mathbf{c}_1$ can win the plurality (TCIS), or maximizing the total number of votes $\mathbf{c}_1$ received (TCMS).
\end{definition}
TCIS was shown to be NP-complete and TCMS to be NP-hard by \citep{bis:article}.
\begin{theorem}\label{thm:MS-hardness}
For 2 candidates, $n$ voters, and $\ell$ binary issues, the problem of maximizing the number of votes for $\mathbf{c}_1$, TCWMS, is NP-hard, even when $B = \infty$ in the normed budget constraint, or when $I = [0, 1]^{\ell}$ in the interval constraint.
\end{theorem}
\begin{proof}
    To prove this claim we will reduce from the problem of 0-1 issue selection on binary issues with two candidates (TCMS).
    An instance of TCMS is defined by a set of $n'$ voters $\pazocal{V}' = \{\mathbf{v}'_1, ..., \mathbf{v}_{n'}'\}$ and two candidates $\mathbf{c}_1', \mathbf{c}_2'$ all of which select positions on $\ell '$ binary issues.
    The objective of TCMS is to maximize the number of votes for $\mathbf{c}_1'$ subject to the constraint that $\mathbf{w}' \in \{0, 1\}^{\ell'}$ and $||\mathbf{w}'||_1 > 0$.
    To reduce from a given instance of TCMS we will add a set of voters that forces any optimal solution to have $\mathbf{w} \in \{0, c\}^\ell$ for some constant $c$ that can be associated with $1$ in the instance of TCMS. Since $B=\infty$, we may assume the adversary is selecting the weight vector $\mathbf{w}$ rather than a perturbation $\mathbf{x}$.
    Without loss of generality we may assume $\mathbf{c}'_1 = \langle 1, 1, ..., 1 \rangle$ and $\mathbf{c'}_2 = \langle 0, 0, ..., 0\rangle$. 
    
    First, let $\ell = 2\ell ' + 2$ and $\mathbf{c}_1 = \langle 1, 1, ..., 1 \rangle, \mathbf{c}_2 = \langle 0, 0, ..., 0 \rangle$. 
    To encode $\pazocal{V}'$, let $\pazocal{V}_1$ be a set of voters obtained by mapping each $\mathbf{v}'_j\in \pazocal{V}'$ to a voter  $\mathbf{v}_j$ where $v_{j, k} = v_{j, k + \ell' + 1} = v'_{j, k}$ for $k \in [\ell ']$ and $v_{j, \ell ' + 1} = 1$, $v_{j, 2\ell ' + 2} = 0$. Compactly, each voter can be represented as $\mathbf{v}_j = \langle v'_{j, 1}, ..., v'_{j, \ell '}, 1, v_{j, 1}, ..., v_{j, \ell '}, 0 \rangle$.
    Next we will introduce five more sets of voters which will force any optimal $\mathbf{w}$ to be binary.
    
    For each $r\in [\ell ' + 1]$ construct $8n'\ell'$ identical copies of a voter who has $v_{j, k} = 1$ when $(k + r) \text{ mod } \ell \leq \ell ' + 1$, and $v_{j, k} = 0$ otherwise. Denote this set of voters as $\pazocal{V}_2$. Let $\pazocal{V}_3$ be a set of voters obtained by taking each $\mathbf{v}_j \in \pazocal{V}_2$ and flipping their opinion, i.e. for each $\mathbf{v}_{j_1}\in \pazocal{V}_2$ add $\mathbf{v}_{j_2}$ to $\pazocal{V}_3$ where $v_{j_2, k} = 1 - v_{j_1, k}$ for all $k\in [\ell]$. Note that $|\pazocal{V}_2| = |\pazocal{V}_3| = 8n'\ell '(\ell ' + 1)$.
    
    Now let $\pazocal{V}_4$ be the set of $4n'\ell '$ voters such that each $\mathbf{v}_j\in\pazocal{V}_4$ has $v_{j, \ell ' + 1} = v_{j, 2\ell ' + 2} = 1$ and $v_{j, k} = 0$ for all other $k$. 
    
    For each $r \in [\ell ']$ we create $2n'$ voters of the form $v_{j, r} = v_{j, r +  \ell ' + 1} = 1$, and for each $k \neq r$, $v_{j, k} =1$,  $v_{j, k + \ell + 1} = 0$. Call this set of $2n'\ell '$ voter $\pazocal{V}_5$.
    Lastly, for each $r \in [\ell ']$ create $2n'$ voters with $v_{j, k} = 1$, $v_{j, k + \ell' + 1} = 0$ for $k \neq r, l + 1$, and $v_{j, r} = v_{j, r + \ell + 1} = 1$, $v_{j, \ell ' + 1} = v_{j, 2\ell ' + 2} = 0$. Call this set of $2n'\ell '$ voters $\pazocal{V}_6$. 
    Let $\pazocal{V} = \pazocal{V}_1 \cup \pazocal{V}_2 \cup \pazocal{V}_3 \cup \pazocal{V}_4 \cup \pazocal{V}_5 \cup \pazocal{V}_6$. 
    
    Note that all voters outside of $\pazocal{V}_1$ have at least $2n'$ copies of themselves. 
    Therefore no optimal solution will have a voter $\mathbf{v}_{j_1} \in \pazocal{V}_1$ vote for $\mathbf{c}_1$ if doing so meant losing any voter $\mathbf{v}_{j_2} \notin \pazocal{V}_1$.
    As a result we will first examine criteria of optimal solutions over $\pazocal{V}_1^c$.
    
    Note that the preference vector of $\mathbf{v}_j$ has $a_{j, k} = 1$ if $v_{j, k} = c_{1, k}$ and $a_{j, k} = -1$ if $v_{j, k} = c_{2, k}$.
    
    Consider the voters in $\pazocal{V}_2$, each of which was created according to some $r \in [\ell ' + 1]$. For each $\mathbf{v}_j \in \pazocal{V}_2$, we have that
    \begin{align*}
        \langle \mathbf{a}_j, \mathbf{w} \rangle \geq 0 \iff \sum_{k_1 \in I_r} w_{k_1} - \sum_{k_2 \in [\ell]\setminus I_r}w_{k_2} \geq 0 
    \end{align*}
    where $I_r = \{k\in[\ell] : ( k + r) \text{ mod } \ell \leq \ell ' + 1\}$.
    Similarly, for each voter $\mathbf{v}_j \in \pazocal{V}_3$ we have 
    \begin{align*}
        \langle \mathbf{a}_j, \mathbf{w} \rangle \geq 0 \iff \sum_{k_1 \in I_r} w_{k_1} - \sum_{k_2 \in [\ell]\setminus I_r}w_{k_2} \leq 0 
    \end{align*}
    
    Therefore, all $16n'(\ell ' + 1)$ voters can be made to vote for $\mathbf{c}_1$ if 
    \begin{align*}
            \sum_{k_1 \in I_r} w_{k_1} = \sum_{k_2 \in [\ell]\setminus I_r}w_{k_2} \quad \forall r \in[\ell ' + 1]
    \end{align*}
    The above system of linear equations has a unique solution, namely $w_{k} = w_{k + \ell' + 1}$ for all $k$. For any $\mathbf{v}_j \in \pazocal{V}_2 \cup \pazocal{V}_3$, there are strictly more copies of $\mathbf{v}_j$ than there are total voters in all other voter sets combined. 
    Therefore any optimal solution must have all voters in $\pazocal{V}_2 \cup \pazocal{V}_3$ voting for $\mathbf{c}_1$. As a result we will work under the assumption that $w_{k} = w_{k + \ell' + 1}$ for all $k$.
    
    All voters in $\pazocal{V}_4$ are of the form $v_{j, r} = v_{j, r + \ell' + 1} = 0$ for some $r \in [\ell ']$, $v_{j, \ell ' + 1} = v_{j, 2\ell ' + 1} = 1$, and $v_{j, k} = 1, v_{j, k + \ell' + 1} = 0$ for all $k \neq r, \ell ' + 1$. If $\pazocal{V}_4$ is made to vote for $\mathbf{c}_1$ then
    \begin{align*}
        w_{\ell ' + 1} + w_{2\ell ' + 1} \geq \max_{k}\{w_k + w_{k + \ell' + 1} : k \in [\ell ']\},
    \end{align*}
    which would immediately imply that 
    \begin{align}\label{cond:1}
        0 < w_{\ell ' + 1} = w_{2\ell ' + 2} \geq \max_{k}\{w_k: k\in[\ell']\}.
    \end{align}
    Since there are more copies of each voter in $\pazocal{V}_4$ than there are total remaining voters and since every voter in $\pazocal{V}_2 \cup \pazocal{V}_3 \cup \pazocal{V}_4$ can be made to vote for $\mathbf{c}_1$, no optimal solution would have any of these voters vote for $\mathbf{c}_2$, and Equation \ref{cond:1} holds.
    
    Finally, consider the voters in $\pazocal{V}_5$ and $\pazocal{V}_6$. Each voter in $\pazocal{V}_5$ is of the form $v_{j, r} = v_{j, r + \ell' + 1} = 0$ for some $r\in [\ell ']$, $v_{j, k} = 1$, $v_{j, k + \ell' + 1} = 0$ for all $k\neq r$. Each voter in $\pazocal{V}_6$ is of the form $v_{j, r} = v_{j, r + \ell' + 1} = 1$ for some $r \in [\ell ']$, $v_{j, \ell ' + 1} = v_{j, 2\ell ' + 2} = 0$, and $v_{j, k} = 1$, $v_{j, k + \ell' + 1} = 0$ for all $k\neq r, \ell ' + 1$. 
    Note that for each $r$ there are $2n'$ copies of the corresponding voter in $\pazocal{V}_5$ and of the corresponding voter in $\pazocal{V}_6$, and that for each $r$ either the set of voters in $\pazocal{V}_5$ vote for $\mathbf{c}_1$ or the voters in $\pazocal{V}_6$ vote for $\mathbf{c}_1$. To see this, fix any $r$ and consider the voters in either set. If the voter from $\pazocal{V}_5$ votes for $\mathbf{c}_1$ then it must be the case that
    \begin{align*}
        \sum_{k \neq \ell ' + 1, 2\ell ' + 2} w_{k}a_{j, k} &\geq w_{\ell ' + 1} + w_{2\ell ' + 2} \\
        \iff w_{r} = w_{r + \ell' + 1}&= w_{\ell ' + 1} = w_{2\ell ' + 2}.
    \end{align*}
    Alternatively, if the voter is in $\pazocal{V}_6$, then 
    \begin{align*}
        &-w_r + w_{r + \ell' + 1} + \sum_{k \neq r}w_k - w_{k + \ell' + 1} \geq 0 \\
        \iff& -w_r - w_{r + \ell' + 1} \geq 0.
    \end{align*}
    Both of the constraints cannot hold since $w_{\ell ' + 1} = 0$ would imply that $\max_{k}\{w_k, k\in[\ell]\} = 0$ and $\sum_{k}^{\ell}w_k = 0$.
    Again, there are $2n'$ copies of each voter in $\pazocal{V}_5, \pazocal{V}_6$ and there remain only $n'$ voters left, so it must be the case that any optimal solution gains either the voter in  $\pazocal{V}_5$ or in $\pazocal{V}_6$ for each $r$.
    Therefore, any optimal solution must have $w_k \in\{0, w_{\ell ' + 1}\}$ with $w_{\ell ' + 1} > 0$. 
    
    Thus, as the only voters left to sway are those in $\pazocal{V}_1$, which corresponds to $\pazocal{V}$, it must be the case that there is a maximum of $8\ell'^2 n' + 13 \ell' n' + \alpha$ voters if and only if an optimal solution in the given instance of TCMS attains $\alpha$ voters. 
\end{proof}

 When there are only two candidates, the problem of winning the plurality vote becomes a special case of maximizing the number of votes for $\mathbf{c}_1$. The proof of Theorem \ref{thm:MS-hardness} can be easily extended to the problem of winning the plurality, by adding a set of voters who agree with $\mathbf{c}_2$ on all issues, such that this set ``cancels out" any votes for $\mathbf{c}_1$ from the constructed voters. These new voters can clearly not be won over by any nonzero weight vector.
 This yields the following theorem.
 
\begin{theorem}\label{cor:plur-hard}
For 2 or more candidates, $n$ voters, and $\ell$ binary issues, the problem of determining if $\mathbf{c}_1$ can win the plurality, TCWP, is NP-complete, even when $B = \infty$ in the normed budget constraint or when $I = [0, 1]^{\ell}$ in the interval constraint..
\end{theorem}

Next, we proceed to considerably strengthen the hardness result in Theorem~\ref{thm:MS-hardness}.
When there are only two candidates and issues are binary, a partial order can be induced on voters by the number of issues they agree with $\mathbf{c}_1$ on. That is, the set of voters can be partitioned into $\ell$ tiers $S_r = \{\mathbf{v}_j\in\pazocal{V}:\sum_{k=1}^{\ell}a_{j, k} = 2r - \ell\}$, where $S_r$ is the set of voters who agree with $\mathbf{c}_1$ on exactly $r$ issues.
We now show that even if for all $r\in[\ell]$ there are only a constant number of voters who agree with $\mathbf{c}_1$ on exactly $r$ issues, maximizing the number of votes for $\mathbf{c}_1$ is still NP-hard.

\begin{theorem}\label{thm:l-voters}
Suppose there are only two candidates, $\ell$ binary issues, and $\Theta(\ell)$ voters.
Suppose further that for each $r\in[\ell]$ $\big|\{ \mathbf{v}_j\in \pazocal{V} : \sum_{k = 1}^{\ell} a_{j, k} = 2r - \ell\big\}|\in \Theta(1)$. Then maximizing the number of voters for $\mathbf{c}_1$ is NP-hard, even when $B = \infty$ in the normed budget constraint, or when $I = [0, 1]^{\ell}$ in the interval constraint.
\end{theorem}

\begin{proof}
    To prove this claim, we will reduce from TCWMS, which was shown to be NP-hard in Theorem~\ref{thm:MS-hardness}. 
    An instance of TCWMS is defined by two candidates $\mathbf{c}_1', \mathbf{c}_2'$, a voter set $\pazocal{V}'$, and a set of $\ell'$ issues taking on values $0, 1$.
    
    In the constructed instance of our problem let $\ell = n'^2\ell'^2$ and assume w.l.o.g.~that $\mathbf{c}_1 = \langle 1, 1, ..., 1 \rangle$, $\mathbf{c}_2 = \langle 0, 0, ..., 0 \rangle$. 
    We will construct a set of voters $\pazocal{V}$ that encodes the voters in $\pazocal{V}'$ such that the election only depends on issues $[\ell']$. 
    To do this, first decompose the set of voters into disjoint sets $\pazocal{V}' = S_1 \cup S_2 \cup ...\cup S_{\ell'}$, such that $S_r = \{\mathbf{v}' \in \pazocal{V}: \sum_{k' = 1}^{\ell'} a'_{k'} = 2r - \ell'\}$. 
    Starting at $r=1$, iterate through each $\mathbf{v'}_j \in S_r$ and create a voter $\mathbf{v}_j$ such that $v_{j, k} = v'_{j, k}$ for all $k\in[\ell']$, $v_{j, k + \ell'} = 1$ for all $k\in [\frac{j^2 + j}{2}: \frac{j^2 + 3j}{2}]$ and $v_{j, k + \ell'} = 0$ otherwise. 
    Under this construction, for any index $s \in[\ell' + 1 : n'^2\ell'^2]$, there is only a single $j$ for which $v_{j, s} = 1$. That is, each voter has either $\sum_{s} w_s$ or $-\sum_{s}w_s$, for $s\in[\ell' +\frac{j^2 + j}{2}: \ell' + \frac{j^2 + 3j}{2}]$, as terms in $\langle \mathbf{w}, \mathbf{a}_j\rangle$. Therefore, each of these index sets $[\ell' +\frac{j^2 + j}{2}: \ell' + \frac{j^2 + 3j}{2}]$ can be associated with a single index, say $j^*$, where for all $j\neq j^*$, $\mathbf{a}_{j, j^*} = -1$ and $\mathbf{a}_{j^*, j^*} = 1$. Under this simplified version of indices, we see that for $\mathbf{v}_j$, and for $k\in[n']$ that the contribution from all $k^*$ issues, to the voters preference sum is, $-w_{1^*} - ... -w_{j^* -1} + w_{j^*} - w_{j^* - 1} - ... -w_{{n'}^*}$. If we take any three of these sums as linear inequalities for voters $\mathbf{v}_{j_1}, \mathbf{v}_{j_2}, \mathbf{v}_{j_3}$, we get 
    \begin{align*}
        -w_{1^*} - ... -w_{{j_1}^* -1} + w_{{j_1}^*} - w_{{j_1}* - 1} - ... -w_{{n'}^*} \geq 0\\
        -w_{1^*} - ... -w_{{j_2}^* -1} + w_{{j_2}^*} - w_{{j_2}^* - 1} - ... -w_{{n'}^*} \geq 0\\
        -w_{1^*} - ... -w_{{j_3}^* -1} + w_{{j_3}^*} - w_{{j_3}^* - 1} - ... -w_{{n'}^*} \geq 0
    \end{align*}
    Since each $0 \leq w_{k^*} \leq 1$, the only satisfying assignment to these three inequalities is $w_{k^*} = 0$ for all $k^*$ . Therefore the objectives of both problems align and this restricted version of WTCMS is NP-hard.
\end{proof}


\subsection{TRACTABLE SPECIAL CASES}
We now return to the setting when there are $m$ candidates, $n$ voters, and issues are real-valued.
Although both \emph{Max Support} and \emph{Majority Vote} are NP-hard even with several strong restrictions, we now show sufficient conditions for either objective to be computed efficiently, as well as algorithms to do so.

Under the normed budget constraint when $p\in\{1, 2, \infty\}$, or under the interval constraint when $p\geq 1$, if the number of voters is $\Theta(\log(\ell))$, or the number of issues is constant and issue values are from a set of constant size, then \emph{Max Support} can be computed in polynomial time. 
For the normed budget constraint with $p\notin\{1, 2, \infty\}$, if the maximum number of votes for $\mathbf{c}_1$ is $\alpha$ when $||\mathbf{x}||_p \leq B$, then for $\epsilon > 0$ a perturbation $\mathbf{x}'$ where $||\mathbf{x}'||_p \leq B + \epsilon$ and $\mathbf{c}_1$ obtains $\alpha' \geq \alpha$ votes, can be found in polynomial time with respect to the input size and $\log\big(\frac{1}{\epsilon}\big)$. Moreover, as $\epsilon \rightarrow 0$, $\alpha'$ asymptotically approaches $\alpha$. 


Under similar assumptions on the number of voters or issues, the objective of \emph{Majority Vote} can be computed in polynomial time for the normed budget constraint with $p\in\{1, 2, \infty\}$, or for the interval constraint with $p\geq 1$. 
In the case of $p\notin\{1, 2, \infty\}$, suppose that $\mathbf{x}^* = \text{argmax}_{\mathbf{x}}\{ ||\mathbf{x}||_p : \mathbf{c_1} \text{ wins the election}\}$, assuming $\mathbf{c}_1$ can be made to win the election. 
For $\epsilon > 0$ a perturbation $\mathbf{x}'$ with $\mathbf{c}_1$ winning the election and $||\mathbf{x}'||_p - ||\mathbf{x}^*|| \leq \epsilon$, can be found in polynomial time with respect to the problem size and $\log\big(\frac{1}{\epsilon}\big)$. 
This $\mathbf{x}'$ might break the attacker's budget by at most $\epsilon$. If $||\mathbf{x}'||_p \leq B$ then, simply taking $\mathbf{x} = \mathbf{x}'$ wins the election within the budget constraint. 
However, in the case when $B < ||\mathbf{x}'||_p \leq B + \epsilon$, it will be unknown whether there exists a $\mathbf{x}$ with $||\mathbf{x}||_p \leq B$ such that $\mathbf{c}_1$ wins the election.
If the attacker is allowed to break their budget by $\epsilon$, i.e. $||\mathbf{x}||_p \leq B + \epsilon$ then taking $\mathbf{x} = \mathbf{x}'$ wins $\mathbf{c}_1$ the election.Further, as $\epsilon \rightarrow 0$, $\mathbf{x}'$ asymptotically approaches $\mathbf{x}^*$.

The existence of polynomial time algorithms for these two objectives is particularly interesting, given that the problem was NP-hard in the case of control by issue selection \emph{even for a single voter}~\citep{bis:article}.

In both cases we use Algorithm \ref{algo:1}, where \emph{unanimity-program}, refers to an optimization program in which all voters in the given demographic, $D\subset \pazocal{V}$, are made to unanimously vote for a given candidate.
\begin{algorithm}\label{algo:1}
\SetAlgoLined
\KwResult{weight vector achieving the most voters}
 \For{$D\in2^{\pazocal{V}}$}{
  Solve unanimity-program over $D$\;
  \If{unanimity-program feasible and is within budget restriction}{
    Store $|D|$ and $\mathbf{w}_D$\;
  }
 }
 \Return argmax $\{|D| : \mathbf{w}_D\}$
 \caption{Maximizing votes for $\mathbf{c}_1$}
\end{algorithm}
Recall that for a given candidate--voter pair $\mathbf{c}_i, \mathbf{v}_j$, the vector $\mathbf{a}_j^{(i)}$ gives $\mathbf{v}_j$'s per-issue preference for $\mathbf{c}_1$ over $\mathbf{c}_i$.

Under the normed budget constraint, the unanimity program for a demographic, $D\subset \pazocal{V}$, is given by
\begin{equation}\label{prog:budget}
\begin{aligned}
        \text{minimize}_{\mathbf{x}}& ||\mathbf{x}||_p\\
        \text{s.t. }& ||\mathbf{x} + \mathbf{w}||_1 = 1 \\
                    & 0 \leq w_k + x_k \leq 1\quad \forall k\in[\ell] \\
                    & \big{\langle} \mathbf{w} + \mathbf{x}, \mathbf{a}_j^{(i)} \big{\rangle} \geq 0\quad \forall i\in[m],~\forall~\mathbf{v}_j \in D
\end{aligned}
\end{equation}
and under the interval constraint, the unanimity program is given by the following linear feasibility problem:
\begin{equation}\label{prog:interval}
\begin{aligned}
        & ||\mathbf{x} + \mathbf{w}||_1 = 1 \\
                    & 0 \leq w_k + x_k \leq 1\quad \forall~k\in[\ell] \\
                    & w_k + x_k \in I_k\quad \forall~k\in[\ell]\\
                    & \big{\langle} \mathbf{w} + \mathbf{x}, \mathbf{a}_j^{(i)} \big{\rangle} \geq 0\quad \forall i\in[m],~\forall~\mathbf{v}_j \in D
\end{aligned}
\end{equation}
\begin{theorem}\label{thm:log(l)-voters}
    Suppose there are $m$ candidates, $\ell$ real-valued issues, and $n\in \Theta(\log(\ell))$ voters and the attacker is restricted by $||\mathbf{x}||_p\leq B$, where $p\in\{1, 2, \infty\}$. Then Algorithm \ref{algo:1} computes \emph{Max Support} for $\mathbf{c}_1$, in polynomial time.
\end{theorem}
\begin{proof}
    Since $|\pazocal{V}| = n \in \Theta(\log(\ell))$, $|2^{\pazocal{V}}| = |2^{\Theta(\log(\ell))}| \in \Theta(\ell)$. Each subset of voters $D\in 2^\pazocal{V}$ is referred to as a demographic. For any $D\in 2^\pazocal{V}$, determining if all voters in $D$ can be made to unanimously vote for $\mathbf{c}_1$ can be computed by solving Program \ref{prog:budget}. 
    We minimize over $||\mathbf{x}||_p$ in order to determine if the minimum change to $\mathbf{w}$, such that all of $D$ votes for $\mathbf{c}_1$, is larger than $B$. That is, demographics that cannot be made to vote for $\mathbf{c}_1$ come in two forms: those where the constraint set is infeasible, and those where the value of the optimal solution is greater than the budget $B$. By selecting the largest viable demographic, the maximum votes for $\mathbf{c}_1$ can be found. For $p\in\{1, 2, \infty\}$ Program \ref{prog:budget} can be solved in polynomial time. When $p=1, \infty$ the program reduces to a liner program, and when $p=2$ the program reduces to a positive definite quadratic program, all of which have polynomial time algorithms.
\end{proof}
\begin{theorem}
    Suppose there are $m$ candidates, $\ell$ real-valued issues, and $n\in \Theta(\log(\ell))$ voters, and the attacker is restricted by $||\mathbf{x}||_p\leq B$, where $p\notin\{1, 2, \infty\}$. Let $\mathbf{x}^*= \text{argmin}_{\mathbf{x}}\{||\mathbf{x}||_p : \mathbf{c_1} \text{ has maximum votes and } ||\mathbf{x}||_p \leq B\}$. Then for any $\epsilon > 0$ a perturbation,  $\mathbf{x}'$, can be computed in polynomial time with respect to the problem size and $\log\big(\frac{1}{\epsilon}\big)$, that obtains at least as many votes as $\mathbf{x}^*$ and $||\mathbf{x}'||_p - ||\mathbf{x}^*||_p \leq \epsilon$.
\end{theorem}
\begin{proof}
    Similarly to Theorem \ref{thm:log(l)-voters}, Program \ref{prog:budget} can be solved for each demographic. In contrast to Theorem \ref{thm:log(l)-voters}, when $p\notin\{1, 2, \infty\}$ we are solving a general convex program, and thus polynomial time solutions will be off by at most a factor of $\epsilon$.
    For a given demographic, $D$, suppose the optimal solution to Program \ref{prog:budget} is $\mathbf{x}_D^*$. 
    Then for $\epsilon > 0$ we can obtain a solution $\mathbf{x}_D'$ such that $||\mathbf{x}_D'||_p \leq ||\mathbf{x}_D^*|| + \epsilon$. 
    As before, demographics that cannot be made to unanimously vote for $\mathbf{c}_1$ come in two forms: demographics in which the constraints of the Program \ref{prog:budget} are infeasible, and demographics for which the optimal $\mathbf{x}_D^*$ has $||\mathbf{x}_D^*||_p > B$. 
    We need not consider demographics of the first type, since the $\epsilon$ approximation of the convex program will not return a vector if the constraint set is infeasible. 
    Via the same strategy as Theorem \ref{thm:log(l)-voters}, we solve each program and take the largest demographic that can be made to unanimously vote for $\mathbf{c}_1$. 
    The key difference in this case, is that we may be selecting a demographic $D'$ that has more voters than than the optimal solution, and requires budget $B + \epsilon$ to obtain. Thus, if $\mathbf{x}^*$ is the smallest vector, with $||\mathbf{x}^*|| \leq B$, that obtains the maximum votes for $\mathbf{c}_1$, then $||\mathbf{x}^*|| \leq ||\mathbf{x}'|| \leq B + \epsilon$.
\end{proof}
\begin{theorem}
    Suppose there are $m$ candidates, $\ell$ real-valued issues, and $n\in \Theta(\log(\ell))$ voters and the attacker has the interval constraint for some interval $I\subset [0, 1]^{\ell}$. Then Algorithm \ref{algo:1} computes \emph{Max Support} for $\mathbf{c}_1$, in polynomial time.
\end{theorem}
\begin{proof}
    Each \emph{unanimity program} is now given by Program \ref{prog:interval}, which is simply a feasibility LP. Therefore determining if a particular demographic can be made to vote for $\mathbf{c}_1$ can be done in polynomial time. As stated in \ref{thm:log(l)-voters}, there are $\theta(\ell)$ demographics that need to be checked and thus the maximum number of votes for $\mathbf{c}_1$ can be computed in polynomial time.
\end{proof}
\begin{theorem}\label{thm:const l}
   Suppose there are $m$ candidates, $n$ voters, and $\ell \in \Theta(1)$ issues, each of which take on values from a set of constant size. Then Algorithm \ref{algo:1} computes \emph{Max Support} for $\mathbf{c_1}$ in polynomial time, for the normed budget restriction when $p\in\{1, 2, \infty\}$.
\end{theorem}
\begin{proof}
Since $\ell \in \Theta(1)$ and positions are selected from a set of constant size, say $r$, then only $r^{\ell}$ distinct voters can exist. 
So, there may be $n$ voters, but at most $\Theta(1)$ of them that need to be investigated. 
Let $\pazocal{V}' \subset\pazocal{V}$ be the set of all unique voters in $\pazocal{V}$. 
For each $\mathbf{v} \in \pazocal{V}'$ we also keep track of the number of times $\mathbf{v}$ appears in $\pazocal{V}$. So, $|2^{\pazocal{V}'}| \in \Theta(1)$ and there are only a constant number of programs to solve, the only difference being that we now choose the feasible demographic representing the maximum number of voters in $\pazocal{V}$, rather than $\pazocal{V}'$. As stated before, each program can be efficiently solved when $p\in\{1, 2, \infty\}$.
\end{proof}
\begin{theorem}
    Suppose there are $m$ candidates, $n$ voters, $\ell \in \Theta(1)$ issues, each of which take on values from a set of constant size, and the attacker is restricted by $||\mathbf{x}||_p\leq B$, where $p\notin\{1, 2, \infty\}$. Let $\mathbf{x}^*= \text{argmin}_{\mathbf{x}}\{||\mathbf{x}||_p : \mathbf{c_1} \text{ has maximum votes and } ||\mathbf{x}||_p \leq B\}$. Then for any $\epsilon > 0$ a perturbation,  $\mathbf{x}'$, can be computed in polynomial time with respect to the problem size and $\log\big(\frac{1}{\epsilon}\big)$, that obtains at least as many votes as $\mathbf{x}^*$ and $||\mathbf{x}'||_p - ||\mathbf{x}^*||_p \leq \epsilon$.
\end{theorem}
\begin{proof}
    We again use the idea in the proof of Theorem \ref{thm:const l} by keeping tack of the unique voters. Once we have the set of unique voters, the proof is identical to that of Theorem \ref{thm:log(l)-voters}.
\end{proof}
\begin{theorem}
Suppose there are $m$ candidates, $n$ voters, and $\ell \in \Theta(1)$ issues, each of which take on values from a set of constant size. Then Algorithm \ref{algo:1} computes \emph{Max Support} for $\mathbf{c_1}$ in polynomial time, under the interval constraint for given intervals $I\subset[0, 1]^{\ell}$. 
\end{theorem}
\begin{proof}
    After constructing the set of unique voters, we solve a constant number of linear programs and take the vector yielding the largest number of votes for $\mathbf{c}_1$.
\end{proof}
\begin{theorem}\label{thm:plur}
    Suppose that there are $m$ candidates and either $n \in \Theta(1)$, or $\ell \in \Theta(1)$ where each issue takes on values from a set of constant size. Then under the budgeted constraint for $p\in\{1, 2, \infty\}$, \emph{Majority Vote} can be computed in polynomial time.
\end{theorem}
\begin{proof}
    In this setting the number of unique partitions of $\pazocal{V}$ is constant. Thus, if there are $m$ candidates, there are $m^{\Theta(1)}$ unique ways in which each partition can be assigned to a candidate. This assignment of disjoint demographics to candidates is equivalent to that particular demographic being made to vote for that candidate. Each pairing, for a given partition $P$, can be given by a set $A = \{(\mathbf{v}_{j_p}, \mathbf{c}_{i_p}), p\in P\}$. To check if there exists a weight vector such that the given pairing is attainable, one need only solve Program \ref{prog:budget}, with the additional set of linear constraints that $\langle \mathbf{w} + \mathbf{x}, \mathbf{a}_{j_p}^{(i_p, i)}\rangle \geq 0$ for all $j\in [n]$, $i\in[m]$, and for all $(\mathbf{v}_{j_p}, \mathbf{c}_{i_p}) \in A$, where $\mathbf{a}_{j_p}^{(i_p, i)}$ is the preference vector of $\mathbf{v}_j$ for $\mathbf{c}_{i_p}$ over $\mathbf{c}_i$. There are only $m^{\Theta(1)}$ programs to solve, each of which takes polynomial time.
\end{proof}
\begin{theorem}
        Suppose that there are $m$ candidates and either $n \in \Theta(1)$, or $\ell \in \Theta(1)$ where each issue takes on values from a set of constant size. 
        Suppose further that the attacker is restricted by $||\mathbf{x}||_p\leq B$, where $p\notin\{1, 2, \infty\}$. Let $\mathbf{x}^*= \text{argmin}_{\mathbf{x}}\{||\mathbf{x}||_p : \mathbf{c_1} \text{ wins and } ||\mathbf{x}||_p \leq B\}$. Then for any $\epsilon > 0$ a perturbation,  $\mathbf{x}'$, 
        can be computed in polynomial time with respect to the problem size and $\log\big(\frac{1}{\epsilon}\big)$, that $\mathbf{c}_1$ wins the election and $||\mathbf{x}'||_p - ||\mathbf{x}^*||_p \leq \epsilon$.
\end{theorem}
\begin{proof}
    As in the proof of Theorem \ref{thm:plur}, the feasibility of each assignment of voters to candidates can be formulated a convex program. An assignment of voters to candidates is valid if the program is feasible and if the optimal solution has value at most $B$. For $p\notin\{1, 2, \infty\}$ these programs cannot be solved exactly in polynomial time, but for any $\epsilon > 0$ where $\frac{1}{\epsilon}$ is polynomial with respect to the problem size, a solution $\mathbf{x}'$, with $||\mathbf{x}'|| - ||\mathbf{x}^*||_p \leq \epsilon$, can be computed efficiently. Thus, we obtain solutions for $\mathbf{x}'$ for each program and take the one with the smallest one under the $l_p$ norm such that $\mathbf{c}_1$ wins the election. 
\end{proof}
\begin{theorem}
    Suppose that there are $m$ candidates and either $n \in \Theta(1)$, or $\ell \in \Theta(1)$ where each issue takes on values from a set of constant size. Then under the interval constraint, for $I\subset[0, 1]^{\ell}$, \emph{Majority Vote} can be computed in polynomial time.
\end{theorem}
\begin{proof}
    Under the interval constraint each possible assignment of voters to candidates can be formulated as a linear program. As shown in the proof of Theorem \ref{thm:plur} there are only a constant number of such assignments and thus we need only solve a constant number of linear programs and then choose the assignment of voters to candidates such that $\mathbf{c}_1$ wins the election.
\end{proof}

\section{STOCHASTIC VOTING}
Another common model for candidate selection is that of stochastic voting, where votes are cast via a distribution over candidates~\cite{Schofeld98}. 
More precisely, let $f$ be a function that maps weighted distance between a voter $\mathbf{v}_j$ and a candidate $\mathbf{c}_i$ to a probability of the voter voting for this candidate.
Next, we show that election control in this setting for general $f$ is NP-hard even when we only have 2 candidates.
For this, suppose that $f$ belongs to a general class of \emph{sigmoidal} functions~\cite{max-sigmoid}, of which the logistic function is a well-known member.
\begin{definition}
    A function $f:[l, u] \rightarrow \mathbb{R}$ is said to be \emph{sigmoidal} if it is Lipshitz continuous and one of the following is true: $f$ is convex, $f$ is concave, or there exists $z \in[l, u]$ such that $f$ is concave on $[l, z]$ and convex on $[z, u]$.
\end{definition}

We now show the hardness of maximizing the expected number of votes for $\mathbf{c}_1$ even in the two-candidate case.

\begin{theorem}
    Suppose there are $2$ or more candidates, $n$ voters and $\ell$ issues, where votes are cast via a sigmoidal function. Then maximizing the expected number of votes for $\mathbf{c}_1$ is NP-hard, even when $B=\infty$, or $I=[0, 1]^{\ell}$.
\end{theorem}

\begin{proof}
    We reduce from the known NP-hard problem Max-2SAT. An instance of Max-2SAT can be defined by a set of $\ell$ Boolean variables $B = \{b_1, ..., b_{\ell}\}$ and a set of $n$ clauses $\Phi = \{(x_{1, 1} \lor x_{1, 2}), ..., (x_{n, 1}\lor x_{n, 2})\}$ where each $x_k\in \{ \lnot b_k, b_k\}$. Let the number of issues be $\ell + 1$, let $\mathbf{c}_1 = \langle 1, 1, ..., 1 \rangle$ and let $\mathbf{c}_2 = \langle 0, 0, ..., 0 \rangle$. Define $\beta_1, \beta_2\in \Theta(n)$. Create $4\ell^2n^2(\beta_1 + \beta_2)$ voters of the form $v_{j, k} = 0$ for all $k\in[\ell]$ and $v_{j, \ell + 1} = 1$.
    For each Boolean variable, $b_r \in B$ create $n^2\beta_1$ voters of the form $v_{j, k} = 0.5$ if $k\notin\{r, \ell + 1\}$, $v_{j, r} = 1$, and $v_{j, \ell + 1} = 0$. 
    Additionally, for each $b_r$, create $n^2\beta_2$ voters of the form $v_{j, k} = 0.5$ for all $k\neq r$ and $v_{j, r} = 0$. Finally, we encode each clause as a voter. Clauses can take on one of the three  forms and we map each form to a voter in the following way:
    \begin{enumerate}
        \item $(b_{r_1}\lor b_{r_2})$ yields $v_{j, r_1} = v_{j, r_2} = 1 - v_{j, \ell + 1} = 0$ and $v_{j, k} = 0.5$ for all $k\neq r_1, r_2, \ell + 1$.
        \item $(\lnot b_{r_1}\lor b_{r_2})$ yields $1 - v_{j, r_1} = v_{j, r_2} = 1$ and $v_{j, k} = 0.5$ for all $k\neq r_1, r_2$.
        \item $(\lnot b_{r_1}\lor \lnot b_{r_2})$ yields $v_{j, r_1} = v_{j, r_2} = 1$, $v_{j, \ell + 1} = 0$, and $v_{j, k} = 0.5$ for all $k \neq r_1, r_2, \ell + 1$.
    \end{enumerate}
    Under this construction, the attacker's objective function can be formulated as 
    \begin{align*}
        & 4n^2\ell^2(\beta_1 + \beta_2)\theta\big(-\sum_{k = 1}^{\ell}w_k + w_{\ell + 1}\big) + \sum_{j = 1}^n\theta\big(\langle \mathbf{w}, \mathbf{a}_j\rangle \big)\\
        + & \sum_{k = 1}^{\ell + 1}\bigg(n^2\beta_1\theta\big(-w_k - w_{\ell + 1}\big) + n^2\beta_2\theta\big( w_k - w_{\ell + 1}\big)\bigg).
    \end{align*}

To complete the proof, we show that this objective is maximized by values that can be mapped back to binary values that satisfy the maximum number of clauses. Since $\theta$ defines a probability, $0 \leq \theta(x) \leq 1$ for any value of $x$.
Suppose $\theta$ is a sigmoid function defined by some sharpness factor $\alpha$.
The objective function can be examined by each of its terms, starting with,
$\theta\big(-\sum_{k = 1}^{\ell}w_k + w_{\ell + 1}\big)$. Since more than $\frac{3}{4}$ of voters contribute to this term, it must be the case that $\theta\big(-\sum_{k = 1}^{\ell}w_k + w_{\ell + 1}\big) \geq \frac{3}{4}$, which implies that $-\sum_{k = 1}^{\ell}w_k + w_{\ell + 1} \geq 0$ as desired. Further, since $0 \leq w_k \leq 1$, it must be the case that $w_{\ell + 1 } \geq \frac{3}{4}$.

Next consider the term $\beta_1 \theta\left(-w_k - w_{\ell + 1}\right) + \beta_2 \theta\left( w_k - w_{\ell + 1}\right) \leq \beta_1 \theta\left(-w_k - \nicefrac{3}{4}\right) + \beta_2\theta\left( w_k - \nicefrac{3}{4}\right)$
Each of these terms is maximized at one of the extremes, $w_k = 0$ or $w_k = 1$. We can always choose the coefficients $\beta_1, \beta_2$ so that the value of the sum agrees at $w_k = 0$ and $w_k = 1$. 
Lastly, we look at the terms that encode the clauses, namely $\theta\big(\langle \mathbf{w}, \mathbf{a}_j\rangle \big)$. Clearly any satisfying assignment of the clause yields a greater value for this term than any non-satisfying assignment. However, it remains to be shown that oversatisfying any number of clauses does not lead to a greater value of the sum than exactly satisfying any single clause. 
We will look at clauses of the form $(b_{r_1}\lor b_{r_2})$, but symmetric analysis holds for the other two cases. Each of these clauses produces a term $\theta\big( w_{r_1} + w_{r_2} - w_{\ell + 1}\big)$. So the marginal gain for over satisfying this clauses is at most
    $\theta\left(1 + 1 - \nicefrac{3}{4}\right) - \theta\left(1 - \nicefrac{3}{4}\right) =  (1 + e^{-\frac{5\alpha}{4}})^{-1} -(1 + e^{-\frac{\alpha}{4}})^{-1}$.
where $\alpha$ is the sharpness of the sigmoid function.
For the sake of analysis, assume $\alpha \geq n$, although similar techniques work for smaller $\alpha$.
Then, the gain from exactly satisfying a single clause is
    $\theta\left(1 - \nicefrac{3}{4}\right) - \theta\left(- \nicefrac{3}{4}\right) =  (1 + e^{-\frac{\alpha}{4}})^{-1} - (1 + e^{\frac{3\alpha}{4}})^{-1}.$

Thus, the gain for exactly satisfying a single clause, compared to over satisfying all $n$ clauses is
     $(1 + e^{-\frac{\alpha}{4}})^{-1} - (1 + e^{\frac{3\alpha}{4}})^{-1} - n ((1 + e^{-\frac{5\alpha}{4}})^{-1} -(1 + e^{-\frac{\alpha}{4}})^{-1}).$
This term is positive for $n \geq 10$. Therefore no optimal solution will over satisfy clauses unless all satisfiable clauses have been satisfied. 
Therefore maximizing the expected number of voters and maximizing the number of satisfied clauses are equivalent. 
\end{proof}

\subsection{A TRACTABLE SPECIAL CASE}

While in general election control is hard even in the stochastic model of voting, we now exhibit a special case which is tractable.
Specifically, we show that for $m$ candidates $n$ voters and $\ell$ real-valued issues, if $f$ is \emph{linear}, the problem of maximizing the expected number of votes for $\mathbf{c}_1$ is in P, under either strength constraint.

\begin{theorem}
    Suppose that $f$ is linear, and there are $m$ candidates, $n$ voters, $\ell$ real-valued issues, and the attacker is restricted by either the normed budget constraint, for any $p\geq 1$, or the interval constraint. Then maximizing the expected number of votes for $\mathbf{c}_1$ is computable in polynomial time. 
\end{theorem}

\begin{proof}
We may more explicitly represent $f(\mathbf{v}_j, \mathbf{c}_i)$ as $f(x_1, ..., x_{m - 1})$, where each variable $x_{i}$ represents $\langle \mathbf{w}, \mathbf{a}_{j}^{(i)} \rangle$. Suppose that $f$ is linear. Since $\mathbf{a}_j^{(i)}$ is a constant value, and the dot product is also linear, $f$ depends linearly on each $w_k$. Therefore maximizing the expected number of voters can be formulated as a convex program
\begin{align*}
    \text{maximize}_{\mathbf{x}}& \sum_{j} f\big(\langle \mathbf{w} + \mathbf{x}, \mathbf{a}_{j}^{(2)} \rangle, ..., \langle \mathbf{w} + \mathbf{x}, \mathbf{a}_{j}^{(m)} \rangle\big)\\
    \text{s.t. }& ||\mathbf{x}||_p \leq B \\
                & 0 \leq w_k + x_k \leq 1\quad \forall k \in[\ell]\\
                & ||\mathbf{w} + \mathbf{x}||_1 = 1
\end{align*}
 In the case of the interval constraint, $||\mathbf{x}||_p\leq B$ is be replaced by $\ell$ linear constraints of the form $w_k + x_k \in I_k$ for intervals $I_k\subset[0, 1]$, making the above program a linear program. Therefore, under the interval constraint, the optimal $\mathbf{x}$ is computable in polynomial time. However, under the budget constraint the program is not linear, but is still solvable in polynomial time. This can be seen by the fact that the objective function is linear with respect to $\mathbf{x}$ and the only non-linear constraint is $||\mathbf{x}||_p\leq B$. As a result, we need only figure out the ratio of increase that each $x_k$ gives to the objective function relative to increase given to $||\mathbf{x}||_p$.
 
 Since $f$ linearly depends on inner products (linear functions) we may decompose the objective sum into 
\begin{align*}
    \bigg(\sum_{k = 1}^{\ell}b_k x_k\bigg) + C  
\end{align*}
for some constant values $b_1, .., b_{\ell}, C$. The constant $C$ can be ignored for the purposes of maximization. First assume that we do not have the constraints $0\leq w_k + x_k \leq 1$ or $||\mathbf{w} + \mathbf{x}||_1 = 1$. Then the program reduces to 
\begin{align}\label{prog:red}
    \text{maximize}_{\mathbf{x}}&\sum_{k = 1}^{\ell}b_k x_k\quad
                \text{s.t. }\quad ||\mathbf{x}||_p \leq B
\end{align}
The optimal solution to this program must have the property that,
    $\big|b_{k_1}x_{k_1}^{-(p-1)}\big| =  \big|b_{k_2}x_{k_2}^{-(p-1)}\big| \quad \forall k_1, k_2\in[\ell]$
Further, any optimal solution must have $||\mathbf{x}||_p = B$, and \text{sign}$(x_k)$ = \text{sign}$(b_k)$.
Therefore the vector maximizing Program \eqref{prog:red} can be computed analytically as
\begin{small}
\begin{align*}
    \mathbf{x}^* = \frac{1}{\sqrt[\leftroot{-2}\uproot{2}p-1]{b_1}}\bigg\langle x_1, \text{sign}(b_2)|x_1|\sqrt[\leftroot{-2}\uproot{2}p-1]{b_2}, ..., \text{sign}(b_{\ell})|x_1|\sqrt[\leftroot{-2}\uproot{2}p-1]{b_{\ell}} \bigg\rangle
\end{align*}
\end{small}
Since $||\mathbf{x}||_p = B$, the value of $x_1$ is unique.

However, this may not be a feasible solution to the original program since both the constraints that $0 \le w_k + x_k \le 1$ and $\sum_{k=1}^{\ell}x_k = 0$ have been ignored. 
The first constraint, $-w_k \le x_k \le 1 - w_k$, may be violated if $x_k$ was made too large or too small in $\mathbf{x}^*$. 
This is be fixed by iteratively computing solutions to Program \ref{prog:red} in the following way. If any $x_k$ does not satisfy the $-w_k \le x_k \le 1 - w_k$, truncate $x_k$ to either value, depending on the sign$(b_k)$. Then remove the truncated variables from the program and re-compute the solution. This must terminate in $\ell$ or fewer steps.

Next we need only deal with the constraint that $\sum_{k=1}^{\ell}x_k = 0$. This constraint can be satisfied by partitioning the $x_k$ values depending on $\text{sign}(x_k)$. Sort each partition the by value of $\big|\frac{b_k}{x_k^{p-1}}\big|$. If $\sum_{k = 1}^{\ell} x_k \neq 0$ then the weight on the $x_k$'s with the smallest ratio of $\big|\frac{b_k}{x_k^{p-1}}\big|$ can be shifted onto $x_k$'s onto elements in the opposite partition with the largest ratio.

Thus both constraints are satisfied without decreasing the objective value, and we have the maximum number expected votes that $\mathbf{c}_1$ can obtain.
\end{proof}

\section{CONCLUSION}

We consider the problem of election control in the spatial model of voting, where a preference of a voter for a candidate is determined by the distance between them in issue space weighted by the relative importance of issues to voters.
We suppose that a malicious actor aims to skew election results in favor of their preferred candidate by changing the relative perceived importance of issues.
We show that this problem is NP-hard for the adversary even when there are only 2 candidates, and whether voters cast their votes deterministically or stochastically.
On the other hand, we exhibit several special cases which are tractable, including settings with a constant number of voters or issues.
Our model of spatial voting opens a novel direction in election control, but still makes a number of limiting assumptions, including an assumption that all voters have the same relative preferences over issues.
Relaxing these is a natural subject for future work.


\subsection*{Acknowledgments}

This research was partially supported by the NSF (IIS-1903207, IIS-1910392), and ARO (W911NF1910241).

\bibliographystyle{plainnat}
\bibliography{uai2020}
\end{document}